\documentclass[conference]{IEEEtran}
\newtheorem{theorem}{Theorem}
\newtheorem{lemma}{Lemma}

\newtheorem{definition}{Definition}
\newtheorem{corollary}{Corollary}

\usepackage{cite}
\usepackage{verbatim}

\usepackage[cmex10]{amsmath}
\usepackage{amssymb}

\usepackage[caption=false,font=footnotesize]{subfig}
\usepackage{graphicx}
\usepackage{xcolor}
\usepackage{rotating}

\usepackage{fixltx2e}
\usepackage{algorithmic,algorithm}
\begin{document}

\title{Information Masking and Amplification: \\ The Source Coding Setting}

\author{\IEEEauthorblockN{Thomas A. Courtade}
\IEEEauthorblockA{Department of Electrical Engineering\\
University of California, Los Angeles\\
Email: tacourta@ee.ucla.edu}}

\maketitle

\begin{abstract}
The complementary problems of masking and amplifying channel state information in the Gel'fand-Pinsker channel have recently been solved by Merhav and Shamai, and Kim et al., respectively.
In this paper, we study a related source coding problem.  Specifically, we consider the two-encoder source coding setting where one source is to be amplified, while the other source is to be masked.  In general, there is a tension between these two objectives which is characterized by the amplification-masking tradeoff.  In this paper, we give a single-letter description of this tradeoff.  

We apply this result, together with a recent theorem by Courtade and Weissman on multiterminal source coding, to solve  a fundamental entropy characterization problem.

\end{abstract}

\section{Introduction}\label{sec:Intro}

The well known source coding with side information problem 
has an achievable rate region given by 
\begin{align*}
R_x \geq H(X|U), ~~
R_y \geq I(Y;U)
\end{align*}
as originally shown by Ahlswede and K\"{o}rner \cite{bib:AhlswedeKorner1975}, and independently by Wyner \cite{bib:Wyner1975}.  In this setting,  the side information encoder merely serves as a helper with the sole purpose of aiding in the recovery of $X^n$ at the decoder.  However, for given rates $(R_x,R_y)$, there may be many different coding schemes which permit recovery of $X^n$ at the decoder. In some cases, it may be desirable to select a coding scheme that reveals very little information about the side information $Y^n$ to the decoder.  We refer to this objective as \emph{masking} the side information.

To motivate this setting, consider the following example.  Suppose $X$ is an attribute of an online customer that an advertiser would like to specifically target (e.g., gender), and $Y$ is other detailed information about the same customer (e.g.,  credit history).  Companies A and B separately have databases $X^n$ and $Y^n$ corresponding to $n$ different customers (the databases could be indexed by IP address, for example).  The advertiser pays Companies A and B to learn as much about the database $X^n$ as possible.  Now, suppose governing laws prohibit the database $Y^n$ from being revealed too extensively.  In this case, the material given to the advertiser must be chosen so that at most a prescribed amount of information is revealed about $Y^n$.  

In general, a masking constraint on $Y^n$ may render near-lossless reconstruction of $X^n$ impossible.  This motivates the study the \emph{amplification-masking tradeoff}.  That is,  the  tradeoff between amplifying (or revealing) information about $X^n$ while simultaneously masking the side information $Y^n$.

Similar problems have been previously considered in the information theory literature on secrecy and privacy.  For example, Sankar et al.~determine the utility-privacy tradeoff for the case of a single encoder in \cite{bib:SankarPrivacy2011}.  In their setting, the random variable $X$ is a vector with a given set of coordinates that should be masked and another set that should be revealed (up to a prescribed distortion). In this context, our study of the amplification-masking tradeoff is a distributed version of \cite{bib:SankarPrivacy2011}, in which utility is measured by the information revealed about  the database $X^n$.  The problem we consider is distinct from those typically studied in the information-theoretic secrecy literature, in that the masking (i.e., equivocation) constraint corresponds to the intended decoder, rather than an eavesdropper.

We remark that the present paper is inspired in part by the recent, complementary works \cite{bib:YHKimSutivongCover2008} and \cite{bib:MerhavShamai2007} which respectively study  amplification and masking of channel state information.  We borrow our terminology from those works.

This paper is organized as follows.  Section \ref{sec:results} formally defines the problems considered and delivers our main results.  The corresponding proofs are given in Section \ref{sec:proofs}.  Final remarks and directions for future work are discussed in Section \ref{sec:conc}.

\section{Problem Statement and Results}\label{sec:results}
Throughout this paper we adopt  notational conventions that are standard in the literature.  Specifically, random variables are denoted by capital letters (e.g., $X$) and their corresponding alphabets are denoted by corresponding calligraphic letters (e.g., $\mathcal{X}$). We abbreviate a sequence $(X_1,\dots,X_n)$ of $n$ random variables by $X^n$,  and we let $\delta(\epsilon)$ represent a quantity satisfying $\lim_{\epsilon\rightarrow 0}\delta(\epsilon)=0$.  
Other notation will be introduced where necessary.

For a joint distribution $p(x,y)$ on finite alphabets $\mathcal{X}\times \mathcal{Y}$, consider the source coding setting where separate Encoders 1 and 2 have access to the sequences $X^n$ and $Y^n$, respectively.  We make the standard assumption that the sequences $(X^n,Y^n)$ are drawn i.i.d.~according to $p(x,y)$ (i.e., $X^n,Y^n \sim \prod_{i=1}^n p(x_i,y_i)$), and $n$ can be taken arbitrarily large.

The first of the following three subsections characterizes the amplification-masking tradeoff.  This result is applied to solve a fundamental entropy characterization in the second subsection.  The final subsection comments on the connection between information amplification and list decoding.   Proofs of the main results are postponed until Section \ref{sec:proofs}.

\subsection{The Amplification-Masking Tradeoff}
Formally, a $(2^{nR_x},2^{nR_y},n)$  code is defined by its encoding functions 
\begin{align*}
f_{x}:\mathcal{X}^n\rightarrow\{1,\dots,2^{nR_x}\} \mbox{~and~}f_{y}:\mathcal{Y}^n\rightarrow\{1,\dots,2^{nR_y}\}. 
\end{align*}
A rate-amplification-masking tuple $(R_x,R_y,\Delta_A, \Delta_M)$ is achievable if, for any $\epsilon>0$, there exists a  $(2^{nR_x},2^{nR_y},n)$ code satisfying the amplification criterion:
\begin{align}
\Delta_A \leq \frac{1}{n}I\left(X^n;f_x(X^n),f_y(Y^n)\right) +\epsilon, \label{eqn:AmpCrit}
\end{align}
and the masking criterion:
\begin{align}
\Delta_M \geq \frac{1}{n}I\left(Y^n;f_x(X^n),f_y(Y^n)\right) -\epsilon. \label{eqn:MaskCrit}
\end{align}
Thus, we see that the amplification-masking problem is an entropy characterization problem similar to that considered in \cite[Chapter 15]{bib:CsiszarKorner1981}.
\begin{definition} The achievable amplification-masking region $\mathcal{R}_{AM}$ is the closure of the set of all achievable rate-amplification-masking tuples $(R_x,R_y,\Delta_A,\Delta_M)$.
\end{definition}

\begin{theorem}\label{thm:Ampmasking}
$\mathcal{R}_{AM}$ consists of the rate-amplification-masking  tuples $(R_x,R_y,\Delta_A,\Delta_M)$ satisfying 
\begin{align}
\hspace{-10pt}
\left. \begin{array}{rl}
R_x  & \geq \Delta_A-I(X;U)\\
R_y &\geq I(Y;U) \\
\Delta_M &\geq \max\left\{  I(Y;U,X)+\Delta_A-H(X), I(Y;U) \right\}\\
\Delta_A &\leq H(X).
\end{array} \right\} \label{eqn:thm1Ineqs}
\end{align}
for some joint distribution $p(x,y,u)=p(x,y)p(u|y)$, where $|\mathcal{U}|\leq|\mathcal{Y}|+1$.
\end{theorem}

Observe that $\mathcal{R}_{AM}$ characterizes the entire tradeoff between amplifying $X^n$ and masking $Y^n$.  We remark that maximum amplification $\Delta_A = H(X)$ does not necessarily imply that $X^n$ can be recovered near-losslessly at the encoder.  However, if an application demands near lossless reproduction of the sequence $X^n$, Theorem \ref{thm:Ampmasking} can be strengthened to include this case.  To this end, define a rate-masking triple $(R_x,R_y,\Delta_M)$ to be achievable if, for any $\epsilon>0$, there exists a  $(2^{nR_x},2^{nR_y},n)$  code satisfying the masking criterion \eqref{eqn:MaskCrit}, and a decoding function
\begin{align*}
\hat{X}^n &: \{1,2,\dots,2^{nR_x}\} \times \{1,2,\dots,2^{nR_y}\}  \rightarrow \mathcal{X}^n
\end{align*}
which satisfies the decoding-error criterion
\begin{align*}
\Pr\left[ X^n \neq \hat{X}^n(f_x(X^n),f_y(Y^n)) \right]\leq \epsilon. 
\end{align*}

\begin{definition} The achievable rate-masking region $\mathcal{R}_M$ is the closure of the set of all achievable rate-masking triples $(R_x,R_y,\Delta_M)$.
\end{definition}

\begin{corollary}\label{cor:masking}
$\mathcal{R}_M$ consists of the rate-masking triples $(R_x,R_y,\Delta_M)$ satisfying 
\begin{align*}
R_x &\geq H(X|U)\\
R_y &\geq I(Y;U)\\
\Delta_M &\geq I(Y;X,U)
\end{align*}
for some joint distribution $p(x,y,u)=p(x,y)p(u|y)$, where $|\mathcal{U}|\leq|\mathcal{Y}|+1$.
\end{corollary}

\subsection{An Entropy Characterization Result}
As we previously noted, the amplification-masking tradeoff solves a multi-letter entropy characterization problem by reducing it to single-letter form.  The reader is directed to \cite{bib:CsiszarKorner1981} for an introduction to entropy characterization problems.  Here, we apply our results to yield a fundamental characterization of the information revealed about $X^n$ and $Y^n$, respectively, by arbitrary encoding functions $f_x$ and $f_y$ (of rates $R_x, R_y$).

\begin{definition}
Define the region $\mathcal{R}^{\star}(R_x,R_y)$ as follows.  The pair $(\Delta_X,\Delta_Y)\in \mathcal{R}^{\star}(R_x,R_y)$ if and only if, for any $\epsilon>0$, there exists a $(2^{nR_x},2^{nRy},n)$ code satisfying
\begin{align*}
\left| \Delta_X -\frac{1}{n}I(X^n;f_x(X^n),f_y(Y^n))\right| &\leq \epsilon, \mbox{~and}\\
\left| \Delta_Y -\frac{1}{n}I(Y^n;f_x(X^n),f_y(Y^n))\right| &\leq \epsilon.
\end{align*}
Let $\overline{\mathcal{R}^{\star}} (R_x,R_y)$ be the closure of $\mathcal{R}^{\star}(R_x,R_y)$.  
\end{definition}

Ultimately we obtain a single-letter description of $\overline{\mathcal{R}^{\star}} (R_x,R_y)$.  However, in order to do so, we require some notation.  To this end, let:
\begin{align*}
\mathcal{R}_{AM}(R_x,R_y) = \left\{ (\Delta_X,\Delta_Y) : (R_x,R_y,\Delta_X,\Delta_Y)\in \mathcal{R}_{AM}  \right\}.
\end{align*}
Symmetrically, let $\mathcal{R}_{MA}$ be the region where $X^n$ is subject to masking $\Delta_X$ and $Y^n$ is subject to amplification $\Delta_Y$.  Let 
\begin{align*}
\mathcal{R}_{MA}(R_x,R_y) = \left\{ (\Delta_X,\Delta_Y) : (R_x,R_y,\Delta_X,\Delta_Y)\in \mathcal{R}_{MA}  \right\}.
\end{align*}
Finally, let $\mathcal{R}_{AA}(R_x,R_y)$ consist of all pairs $(\Delta_X,\Delta_Y)$ satisfying
\begin{align*}
R_x  & \geq I(U_x;X|U_y,Q)\\
R_y &\geq I(U_y;Y|U_x,Q) \\
R_x+R_y &\geq I(U_x,U_y;X,Y|Q) \\
\Delta_X &\leq I(X;U_x,U_y|Q) \\
\Delta_Y &\leq  I(Y;U_x,U_y|Q)
\end{align*}
for some joint distribution of the form
\begin{align*}
p(x,y,u_x,u_y,q)=p(x,y)p(u_x|x,q)p(u_y|y,q)p(q), 
\end{align*}
where $|\mathcal{U}_x|\leq|\mathcal{X}|$, $|\mathcal{U}_y|\leq|\mathcal{Y}|$, and $|\mathcal{Q}|\leq5$.

\begin{theorem}\label{thm:entropy}
The region $\overline{\mathcal{R}^{\star}} (R_x,R_y)$ has a single-letter characterization given by
\begin{align*}
\overline{\mathcal{R}^{\star}}& (R_x,R_y)=\\
&\mathcal{R}_{AM}(R_x,R_y)  \cap \mathcal{R}_{MA}(R_x,R_y)  \cap \mathcal{R}_{AA}(R_x,R_y) .
\end{align*}
Moreover, restriction of the encoding functions to vector-quantization and/or random binning is sufficient to achieve any point in $\overline{\mathcal{R}^{\star}} (R_x,R_y)$.
\end{theorem}

The second statement of Theorem \ref{thm:entropy} is notable since it states that relatively simple encoding functions (i.e., vector quantization and/or binning) can asymptotically reveal the same amount of information about $X^n$ and $Y^n$, respectively, as encoding functions that are only restricted in rate.  In contrast, this is not true for the setting of three or more sources, as the modulo-sum problem studied by K\"{o}rner and Marton \cite{bib:KornerMarton1979} provides a counterexample where the Berger-Tung achievability scheme  \cite{bib:BergerLongo1977} is not optimal.  Thus, obtaining a characterization like Theorem \ref{thm:entropy} for three or more sources represents a formidable challenge.

We remark that the points in $\overline{\mathcal{R}^{\star}} (R_x,R_y)$ with $\Delta_X=H(X)$ and/or $\Delta_Y=H(Y)$ also capture the more stringent constraint(s) of near-lossless reproduction of $X^n$ and/or $Y^n$, respectively.  This is a consequence of  Corollary \ref{cor:masking}.

To give a concrete example of $\overline{\mathcal{R}^{\star}} (R_x,R_y)$, consider the following joint distribution:
\begin{align}
\begin{tabular}{c|cc}
$P_{X,Y}(x,y)$ & $x=0$ & $x=1$\\ \hline
$y=0$ & $1/3$ & $0$ \\
$y=1$ & $1/6$ & $1/2.$ \\
\end{tabular} \label{eqn:jointDist}
\end{align}
By performing a brute-force search over the auxiliary random variables defining $\overline{\mathcal{R}^{\star}} (R_x,R_y)$ for the distribution $P_{X,Y}$, we have  obtained numerical approximations of $\overline{\mathcal{R}^{\star}} (\cdot,\cdot)$ for several different pairs of $(R_x,R_y)$.  The results are given in Figure \ref{fig:regionExp}.

\begin{figure}
\centering
\includegraphics[trim = 28mm 72mm 30mm 75mm, clip,width=3.25in]{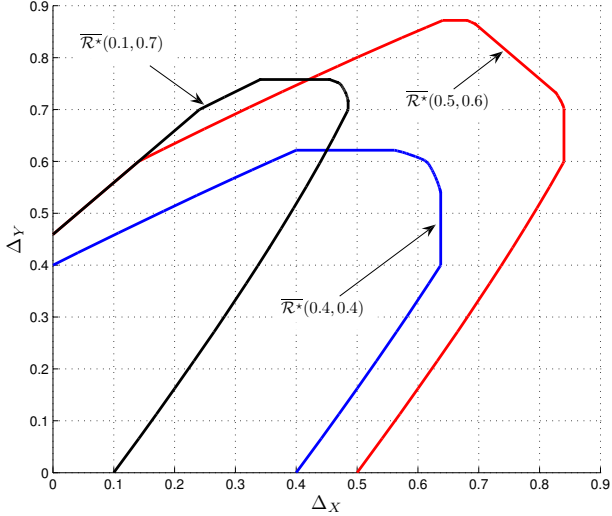}
\caption{The region $\overline{\mathcal{R}^{\star}} (R_x,R_y)$ for joint distribution $P_{X,Y}$ given by \eqref{eqn:jointDist} and three different pairs of rates. Rate pairs $(R_x,R_y)$ equal to $(0.1,0.7)$, $(0.4,0.4)$, and $(0.5,0.6)$ define the convex regions bounded by the black, blue, and red curves, respectively.}
\label{fig:regionExp}
\end{figure}

\subsection{Connection to List Decoding}
We briefly comment on the connection between an amplification constraint and list decoding.
As discussed in detail in \cite{bib:YHKimSutivongCover2008}, the amplification criterion \eqref{eqn:AmpCrit} is essentially equivalent to the requirement for a list decoder
\begin{align*}
L_n : \{1,\dots,2^{nR_x}\}\times \{1,\dots,2^{nR_y}\}\rightarrow 2^{\mathcal{X}^n}
\end{align*}
with list size and probability of error respectively satisfying
\begin{align*}
& \log |L_n| \leq n(H(X)-\Delta_A+\epsilon), \mbox{and}\\
& \Pr\left[ X^n \notin L_n(f_x(X^n),f_y(Y^n)) \right]\leq \epsilon.
\end{align*}
Thus maximizing the amplification of $X^n$ subject to given rate and masking constraints can be thought of as characterizing the best list decoder in that setting.

\section{Proofs of Main Results}\label{sec:proofs}

\begin{proof}[Proof of Theorem \ref{thm:Ampmasking}]
\emph{Converse Part:} Suppose $(R_x,R_y,\Delta_A,\Delta_M)$ is achievable. For convenience, define $F_x=f_x(X^n)$, $F_y=f_y(Y^n)$, and $U_i=(F_y,Y^{i-1})$.

First, note that $\Delta_A \leq H(X)$ is trivially satisfied.  Next, 
the constraint on $R_x$ is given by:
\begin{align}
nR_x &\geq H(F_x) \geq H(F_x|F_y)\notag \\
 &=\sum_{i=1}^n H(X_i|F_y,X^{i-1}) -H(X^n|F_x,F_y)\notag\\
 &\geq \sum_{i=1}^n H(X_i|F_y,Y^{i-1},X^{i-1}) -H(X^n|F_x,F_y) \notag \\
 &=I(X^n;F_x,F_y)-\sum_{i=1}^n I(X_i;U_i)\label{eqn:markovAM} \\
 &\geq n(\Delta_A-\epsilon) -\sum_{i=1}^n I(X_i;U_i).\label{eqn:convAmpCrit}
\end{align}
Equality \eqref{eqn:markovAM} follows since $X_i\leftrightarrow F_y,Y^{i-1}\leftrightarrow X^{i-1}$ form a Markov chain, and inequality \eqref{eqn:convAmpCrit} follows since amplification $\Delta_A$ is achievable.

The constraint on $R_y$ is trivial:
\begin{align*}
nR_y &\geq H(F_y) \geq I(F_y;Y^n)=\sum_{i=1}^n I(Y_i;F_y|Y^{i-1})\\
&=\sum_{i=1}^n I(Y_i;F_y,Y^{i-1})=\sum_{i=1}^n I(Y_i;U_i).
\end{align*}
Similarly, we obtain the first lower bound on $\Delta_M$:
\begin{align*}
n(\Delta_M+\epsilon) &\geq I(Y^n;F_x,F_y) 
\geq  I(Y^n;F_y) = \sum_{i=1}^n I(Y_i;U_i).
\end{align*}
The second lower bound on $\Delta_M$ requires slightly more work, and can be derived as follows:
\begin{align}
&n (\Delta_M+\epsilon) \geq I(Y^n;F_x,F_y) \notag\\
&=I(Y^n;X^n,F_y)+I(X^n;F_x,F_y)-I(X^n;F_x,Y^n) \notag \\
&\geq I(Y^n;X^n,F_y)+n\Delta_A-I(X^n;F_x,Y^n)-n\epsilon \label{eqn:convAMAmp}\\
&\geq \sum_{i=1}^n I(Y_i;X^n,F_y|Y^{i-1})+n\Delta_A-H(X^n)-n\epsilon\notag \\
&\geq \sum_{i=1}^n I(Y_i;X_i,U_i)+\Delta_A-H(X_i)-\epsilon,\notag
\end{align}
where  \eqref{eqn:convAMAmp} follows since amplification $\Delta_A$ is achievable.

Observing that the Markov condition $U_i \leftrightarrow Y_i \leftrightarrow X_i$ is satisfied for each $i$, a standard timesharing argument proves the existence of a random variable $U$ such that $U \leftrightarrow Y \leftrightarrow X$ forms a Markov chain and \eqref{eqn:thm1Ineqs} is satisfied.

\emph{Direct Part:}
Fix $p(u|y)$ and suppose $(R_x,R_y,\Delta_A,\Delta_M)$ satisfy \eqref{eqn:thm1Ineqs} with strict inequality.
Next, fix $\epsilon>0$ sufficiently small so that it is less than the minimum slack in said inequalities, and set $\tilde{R}=I(Y;U)+\epsilon$.  Our achievability scheme uses a standard random coding argument which we sketch below.

\textbf{Codebook generation.} Randomly and independently, bin the typical $x^n$'s uniformly into $2^{n(\Delta_A-I(X;U)+\epsilon)}$ bins.  Let $b(x^n)$ be the index of the bin which contains $x^n$.  For $l\in\{1,\dots,2^{n\tilde{R}}\}$, randomly and independently generate $u^n(l)$, each according to $\prod_{i=1}^n p_U(u_i)$.  

\textbf{Encoding.} Encoder 1, upon observing the sequence $X^n$, sends the corresponding bin index $b(X^n)$ to the decoder.  If $X^n$ is not typical, an error is declared.  Encoder 2, upon observing the sequence $Y^n$, finds an $L\in\{1,\dots,2^{n\tilde{R}}\}$ such that $(Y^n,U^n(L))$ are jointly typical, and sends the unique index $L$ to the decoder.  If more than one such $L$ exists, ties are broken arbitrarily. If no such $L$ exists, then an error is declared.

This coding scheme clearly satisfies the given rates.  Further, each encoder errs with arbitrarily small probability as $n\rightarrow \infty$.  Hence, we only need to check that the amplification and masking constraints are satisfied.  To this end, let $\mathcal{C}$ be the random codebook.  We first check that the amplification and masking constraints are separately satisfied when averaged over random codebooks $\mathcal{C}$.

To see that the (averaged) amplification constraint is satisfied, consider the following:
\begin{align}
&I(X^n; F_x,F_y|\mathcal{C})
=H(X^n|\mathcal{C})-H(X^n|b(X^n),L,\mathcal{C})\notag\\
&\geq nH(X)-n(H(X)-\Delta_A+\delta(\epsilon)) \label{eqn:numXtyp}\\
&=n (\Delta_A-\delta(\epsilon)),\notag
\end{align}
where \eqref{eqn:numXtyp} follows since $X^n$ is independent of $\mathcal{C}$ and, averaged over codebooks, there are at most $2^{n(H(X)-\Delta_A+\delta(\epsilon))}$ sequences $x^n$ in bin $b(X^n)$ which are typical with $U^n(L)$, where $L\in\{1,\dots,2^{n\tilde{R}}\}$.  The details are given in the Appendix.

We now turn our attention to the masking criterion.  First note the following inequality:
\begin{align} 
&I(Y^n;F_x,F_y|\mathcal{C}) 
=I(Y^n;L|\mathcal{C})+I(Y^n;b(X^n)|L,\mathcal{C}) \notag\\
&\leq I(Y^n;L|\mathcal{C})+H(b(X^n)|\mathcal{C}) -H(b(X^n)|Y^n,\mathcal{C})\notag\\
&=I(Y^n;L|\mathcal{C})+I(X^n;Y^n)-H(X^n)+H(b(X^n)|\mathcal{C}) \notag \\
&\quad -H(b(X^n)|Y^n,\mathcal{C})+H(X^n|Y^n) \notag\\
&\leq I(Y^n;L|\mathcal{C})+I(X^n;Y^n)-H(X^n)+H(b(X^n)|\mathcal{C}) \notag \\
&\quad -I(b(X^n);X^n|Y^n,\mathcal{C})+H(X^n|Y^n) \notag\\
&=I(Y^n;L|\mathcal{C})+I(X^n;Y^n)-H(X^n)+H(b(X^n)|\mathcal{C}) \notag \\
&\quad+H(X^n|Y^n,b(X^n),\mathcal{C})\label{eqn:continue}
\end{align}
Two of the terms in \eqref{eqn:continue} can be bounded as follows:  First, since $L \in\{1,\dots,2^{n\tilde{R}}\}$, we have
\begin{align*}
I(Y^n;L|\mathcal{C})\leq  n\tilde{R} =  n(I(Y;U)+\epsilon).
\end{align*}
Second,  there are $2^{n(\Delta_A-I(X;U)+\epsilon)}$ bins at Encoder $1$ by construction, and hence $H(b(X^n)|\mathcal{C})\leq n(\Delta_A -I(X;U)+\epsilon)$.
Therefore, substituting into \eqref{eqn:continue} and simplifying, we have:
\begin{align} 
I(Y^n;F_x,F_y|\mathcal{C}) 
&\leq n(I(Y;U,X)+\Delta_A-H(X)) \notag \\
&\quad+H(X^n|Y^n,b(X^n),\mathcal{C}) +n2\epsilon. \label{eqn:finalStep}
\end{align}
We now consider three separate cases.  First, assume $\Delta_A\leq I(U;X)$.
Then, 
\begin{align*}
I(Y;X,U)+\Delta_A-H(X) &\leq I(Y;X,U)-H(X|U) \\
&= I(Y;U)-H(X|Y),
\end{align*}
and \eqref{eqn:finalStep} becomes 
\begin{align*}
I(Y^n;F_x,F_y|\mathcal{C}) 
&\leq nI(Y;U) -I(X^n;b(X^n)|Y^n,\mathcal{C}) +n2\epsilon\\
&\leq nI(Y;U) +n 2 \epsilon.
\end{align*}
Next, suppose that $\Delta_A \geq I(X;U) + H(X|Y)$.
In this case, there are greater than $2^{n(H(X|Y)+\epsilon)}$ bins in which the $X^n$ sequences are distributed.  Hence, knowing $Y^n$ and $b(X^n)$ is sufficient to determine $X^n$ with high probability (i.e., we have a Slepian-Wolf binning at Encoder 1).  Therefore, $H(X^n|Y^n,b(X^n),\mathcal{C})\leq n\epsilon$, and \eqref{eqn:finalStep} becomes 
\begin{align*}
I(Y^n;F_x,F_y|\mathcal{C}) 
\leq n(I(Y;X,U)+\Delta_A-H(X))  +n3\epsilon.
\end{align*}

Finally, suppose $\Delta_A = I(X;U) + \theta H(X|Y)$ for some $\theta\in [0,1]$.  In this case, we can timeshare between a code $\mathcal{C}_1$ designed for amplification $\Delta_A'=I(X;U)$ with probability $\theta$, and a code $\mathcal{C}_2$ designed for amplification $\Delta_A''=I(X;U)+H(X|Y)$ with probability $1-\theta$ to obtain a code $\mathcal{C}$ with the same average rates and averaged amplification
\begin{align*}
&I(X^n;F_x,F_y|\mathcal{C}) \\
&= \theta I(X^n;F_x,F_y|\mathcal{C}_1) + (1-\theta)I(X^n;F_x,F_y|\mathcal{C}_2)\\
&\geq n(I(X;U) + \theta H(X|Y)-\delta(\epsilon)) = n(\Delta_A-\delta(\epsilon)).
\end{align*}  Then, applying the inequalities obtained in the previous two cases, we obtain:
\begin{align*}
&I(Y^n;F_x,F_y|\mathcal{C})\\
&=\theta I(Y^n;F_x,F_y|\mathcal{C}_1)+(1-\theta) I(Y^n;F_x,F_y|\mathcal{C}_2) \\
&\leq \theta nI(Y;U) +(1-\theta)n(I(Y;X,U)+\Delta_A''-H(X)) + 3n \epsilon\\
&=  nI(Y;U) + 3n \epsilon.
\end{align*}
Combining these three cases proves that 
\begin{align*}
&\frac{1}{n}I(Y^n;F_x,F_y|\mathcal{C})\\
&\quad \leq \max\{I(Y;U,X)+\Delta_A-H(X),I(Y;U)\}+3\epsilon \\
&\quad \leq \Delta_M+3\epsilon.
\end{align*}
To show that there exists a code which satisfies the amplification and masking constraints simultaneously, we construct a super-code $\bar{\mathcal{C}}$ of blocklength $Nn$ by concatenating $N$ randomly, independently chosen codes of length $n$ (each constructed as described above).  By the weak law of large numbers and independence of the concatenated coded blocks,
\begin{align*}
\Pr\left(\left\{ \bar{c} : \frac{1}{Nn}I(X^{Nn};\bar{F}_x,\bar{F}_y|\bar{\mathcal{C}}=\bar{c})>\Delta_A- \delta(\epsilon) \right\}\right) & \geq 3/4\\
\Pr\left(\left\{ \bar{c} : \frac{1}{Nn}I(Y^{Nn};\bar{F}_x,\bar{F}_y|\bar{\mathcal{C}}=\bar{c})<\Delta_M+ \delta(\epsilon)\right\} \right) & \geq 3/4
\end{align*}
for $N$ and $n$ sufficiently large.  Thus, there must exist one super-code which simultaneously satisfies both desired constraints.
This completes the proof that $(R_x,R_y,\Delta_A,\Delta_M)$ is achievable.  Finally,  we invoke the Support Lemma \cite{bib:CsiszarKorner1981} to see that $|\mathcal{Y}|-1$ letters are sufficient to preserve $p(y)$.  Plus, we require two more letters to preserve the values of $H(X|U)$ and $I(Y;U|X)$.
\end{proof}

\begin{proof}[Proof of Corollary \ref{cor:masking}]
By setting $\Delta_A=H(X)$, \cite[Theorem 2]{bib:AhlswedeKorner1975} implies that $X^n$ can be reproduced near losslessly.  A simplified version of the argument in the direct part of the proof of Theorem \ref{thm:Ampmasking} shows that the masking criterion will be satisfied for the standard coding scheme.  The converse of Theorem \ref{thm:Ampmasking} continues to apply
\end{proof}

\begin{proof}[Proof of Theorem \ref{thm:entropy}]
First, we remark that the strengthened version of \cite[Theorem 6]{bib:CourtadeWeissman2011} states that $\mathcal{R}_{AA}(R_x,R_y)$ is the closure of pairs $(\Delta_X,\Delta_Y)$ such that there exists a $(2^{nR_x}, 2^{nR_y}, n)$ code satisfying 
\begin{align*}
\Delta_X &\leq \frac{1}{n}I(X^n;f_x(X^n),f_y(Y^n)) + \epsilon,\\
\Delta_Y &\leq \frac{1}{n}I(Y^n;f_x(X^n),f_y(Y^n)) + \epsilon
\end{align*}
for any $\epsilon>0$.

Suppose $(\Delta_X,\Delta_Y)\in \mathcal{R}^{\star}(R_x,R_y)$. By definition of $ \mathcal{R}^{\star}(R_x,R_y)$, Theorem \ref{thm:Ampmasking}, and the above statement, $(\Delta_X,\Delta_Y)$ also lies in each of the sets $\mathcal{R}_{AM}(R_x,R_y)$,  $\mathcal{R}_{MA}(R_x,R_y)$, and $\mathcal{R}_{AA}(R_x,R_y)$.  Since each of these sets are closed by definition, we must have 
\begin{align*}
\overline{\mathcal{R}^{\star}}& (R_x,R_y)\subseteq \\
&\mathcal{R}_{AM}(R_x,R_y)  \cap \mathcal{R}_{MA}(R_x,R_y)  \cap \mathcal{R}_{AA}(R_x,R_y) .
\end{align*}

Since each point in the sets $\mathcal{R}_{AM}(R_x,R_y)$,  $\mathcal{R}_{MA}(R_x,R_y)$, and $\mathcal{R}_{AA}(R_x,R_y)$ is achievable by vector quantization and/or random binning, the second statement of the Theorem is proved.

To show the reverse inclusion, fix $\epsilon>0$ and suppose $(\Delta_X,\Delta_Y)\in \mathcal{R}_{AM}(R_x,R_y)  \cap \mathcal{R}_{MA}(R_x,R_y)  \cap \mathcal{R}_{AA}(R_x,R_y)$.  This implies the existence of $(2^{n_{AM}R_x}, 2^{n_{AM}R_y}, n_{AM})$, $(2^{n_{MA}R_x}, 2^{n_{MA}R_y}, n_{MA})$, and $(2^{n_{AA}R_x}, 2^{n_{AA}R_y}, n_{AA})$ codes satisfying:
\begin{align*}
\Delta_X &\leq \frac{1}{n_{AM}}I(X^{n_{AM}};f_x^{AM}(X^{n_{AM}}),f_y^{AM}(Y^{n_{AM}})) + \epsilon,\\
\Delta_Y &\geq \frac{1}{n_{AM}}I(Y^{n_{AM}};f_x^{AM}(X^{n_{AM}}),f_y^{AM}(Y^{n_{AM}})) - \epsilon.\\
\Delta_X &\geq \frac{1}{n_{MA}}I(X^{n_{MA}};f_x^{MA}(X^{n_{MA}}),f_y^{MA}(Y^{n_{MA}})) - \epsilon,\\
\Delta_Y &\leq \frac{1}{n_{MA}}I(Y^{n_{MA}};f_x^{MA}(X^{n_{MA}}),f_y^{MA}(Y^{n_{MA}})) + \epsilon,\\
\Delta_X &\leq \frac{1}{n_{AA}}I(X^{n_{AA}};f_x^{AA}(X^{n_{AA}}),f_y^{AA}(Y^{n_{AA}})) + \epsilon,\\
\Delta_Y &\leq \frac{1}{n_{AA}}I(Y^{n_{AA}};f_x^{AA}(X^{n_{AA}}),f_y^{AA}(Y^{n_{AA}})) + \epsilon.
\end{align*}
Also, by taking $f_x^{MM}, f_y^{MM}$ to be constants, we trivially have a $(2^{n_{MM}R_x}, 2^{n_{MM}R_y}, n_{MM})$ code satisfying
\begin{align*}
\Delta_X &\geq \frac{1}{n_{MM}}I(X^{n_{MM}};f_x^{MM}(X^{n_{MM}}),f_y^{MM}(Y^{n_{MM}})) ,\\
\Delta_Y &\geq \frac{1}{n_{MM}}I(Y^{n_{MM}};f_x^{MM}(X^{n_{MM}}),f_y^{MM}(Y^{n_{MM}})).
\end{align*}
It is readily verified that, by an appropriate timesharing between these four codes, there exists a $(2^{nR_x}, 2^{nR_y}, n)$ code satisfying \begin{align*}
\left| \Delta_X -\frac{1}{n}I(X^n;f_x(X^n),f_y(Y^n))\right| &\leq \delta(\epsilon), \mbox{~and}\\
\left| \Delta_Y -\frac{1}{n}I(Y^n;f_x(X^n),f_y(Y^n))\right| &\leq \delta(\epsilon).
\end{align*}
This completes the proof of the theorem.
\end{proof}

\section{Concluding Remarks}\label{sec:conc}
In this paper, we considered a setting where two separate encoders have access to correlated sources.  We gave a complete characterization of the tradeoff between amplifying information about one source while simultaneously masking another.  By combining this result with recent results by Courtade and Weissman \cite{bib:CourtadeWeissman2011}, we precisely characterized the amount of information that can be revealed about $X^n$ and $Y^n$ by any encoding functions satisfying given rates.  There are three notable points here: (i) this multi-letter entropy characterization problem admits a single-letter solution,  (ii) restriction of encoding functions to vector quantization and/or random binning is sufficient to achieve any point the region, and (iii) this simple characterization does not extend to three or more sources/encoders.

Finally, we remark that in the state amplification and masking problems considered in \cite{bib:YHKimSutivongCover2008} and \cite{bib:MerhavShamai2007}, the authors obtain explicit characterizations of the achievable regions when the channel state and noise are independent Gaussian random variables.  Presumably, this could also be accomplished in our setting using known results on Gaussian multiterminal source coding, however, a compete investigation into this matter is beyond the scope of this paper

\section*{Acknowledgment}
The author gratefully acknowledges the conversations with Tsachy Weissman and comments by an anonymous reviewer which contributed to this paper. 

\appendix
\begin{lemma}
With all quantities defined as in the proof of Theorem \ref{thm:Ampmasking},
\begin{align*}
\limsup_{n\rightarrow\infty} \frac{1}{n} H(X^n|L,b(X^n),\mathcal{C}) \leq H(X)-\Delta_A+\delta(\epsilon).
\end{align*}
\end{lemma}
\begin{proof}
We follow the proof strategy of \cite[Lemma 22.3]{bib:ElGamalYHKim2012} and make adjustments where necessary.  
For convenience, define $\tilde{R}_x=\Delta_A - I(X;U)+\epsilon$ and recall that $\epsilon$ was chosen sufficiently small so that $\tilde{R}_x<H(X|U)$.  Note that we can express the random codebook $\mathcal{C}$ as a pair of random codebooks $\mathcal{C}=(\mathcal{C}_B,\mathcal{C}_{VQ})$, where $\mathcal{C}_B$ is the ``binning codebook" at Encoder 1, and $\mathcal{C}_{VQ}$ is the ``vector-quantization codebook" at Encoder 2.

Let $E_1=1$ if $(X^n,U^n(L))\notin \mathcal{T}_{\epsilon}^{(n)}$ and $E_1=0$ otherwise.  Note that $\Pr(\{E_1=1\})$ tends to $0$ as $n\rightarrow\infty$.
Consider
\begin{align*}
&H(X^n|L,b(X^n),\mathcal{C})\\
&\leq H(X^n,E_1|L,b(X^n),\mathcal{C})\\
&\leq 1+n\Pr(\{E_1=1\})H(X)\\
&+\sum_{(l,b,c_{VQ})}p(l,b,c_{VQ}|E_1=0)\\
&\quad \times H(X^n|L=l,b(X^n)=b,E_1=0,\mathcal{C}_{VQ}=c_{VQ},\mathcal{C}_{B}).
\end{align*}
Now, let $N(l,b,c_{VQ},\mathcal{C}_B)$ be the number of sequences $x^n\in \mathcal{B}(b)\cap \mathcal{T}_{\epsilon}^{(n)}(X|u^n(l))$, where $\mathcal{B}(b)$ denotes the bin of $x$-sequences which is labeled by index $b$ and $u^n(l)$ is the codeword in the (fixed) codebook $c_{VQ}$ with index $l$.  Note that $N(l,b,c_{VQ},\mathcal{C}_B)$ is a binomial random variable, where the source of randomness comes from the random codebook $\mathcal{C}_B$. Define
\begin{align*}
&E_2(l,b,c_{VQ},\mathcal{C}_B) \\
&\quad = \left\{ \begin{array}{ll} 1 & \mbox{if $N(l,b,c_{VQ},\mathcal{C}_B)\geq 2\mathbb{E}\left[N(l,b,c_{VQ},\mathcal{C}_B)\right]$,}\\
0 & \mbox{otherwise.}
\end{array}\right.
\end{align*}
Due to the binomial distribution of $N(l,b,c_{VQ},\mathcal{C}_B)$, it is readily verified that
\begin{align*}
\mathbb{E} \left[N(l,b,c_{VQ},\mathcal{C}_B)\right] &= 2^{-n\tilde{R}_x}\left|\mathcal{T}_{\epsilon}^{(n)}(X|u^n(l))\right|,\\
\mbox{Var}(N(l,b,c_{VQ},\mathcal{C}_B)) &\leq 2^{-n\tilde{R}_x}\left|\mathcal{T}_{\epsilon}^{(n)}(X|u^n(l))\right| .
\end{align*}
Then, by the Chebyshev lemma \cite[Appendix B]{bib:ElGamalYHKim2012},
\begin{align*}
\Pr(\{E_2(l,b,c_{VQ},\mathcal{C}_B) =1\})
&\leq \frac{\mbox{Var}(N(l,b,c_{VQ},\mathcal{C}_B))}{\left( \mathbb{E} \left[ N(l,b,c_{VQ},\mathcal{C}_B) \right] \right)^2}\\
&\leq 2^{-n(H(X|U)-\tilde{R}_x-\delta(\epsilon))},
\end{align*}
which tends to zero as $n\rightarrow\infty$ if $\tilde{R}_x<H(X|U)-\delta(\epsilon)$, which is satisfied for $\epsilon$ sufficiently small.  Now consider
\begin{align*}
&H(X^n|L=l,b(X^n)=b,E_1=0,\mathcal{C}_{VQ}=c_{VQ},\mathcal{C}_{B}) \\
&\leq H(X^n,E_2|L=l,b(X^n)=b,E_1=0,\mathcal{C}_{VQ}=c_{VQ},\mathcal{C}_{B}) \\
&\leq 1 + n\Pr(\{E_2=1\})H(X) \\
&+ H(X^n|L=l,b(X^n)=b,E_1=0,E_2=0,\mathcal{C}_{VQ}=c_{VQ},\mathcal{C}_{B}) \\
&\leq 1 + n\Pr(\{E_2=1\})H(X) \\
&+ n(H(X|U)-\tilde{R}_x+\delta(\epsilon)), 
\end{align*}
which implies that
\begin{align*}
&H(X^n|L,b(X^n),\mathcal{C})\\
&\leq 2+n(\Pr(\{E_1=1\})+\Pr(\{E_2=1\}))H(X) \\
&\quad+ n(H(X|U)-\tilde{R}_x+\delta(\epsilon))\\
&\leq 2+n(\Pr(\{E_1=1\})+\Pr(\{E_2=1\}))H(X) \\
&\quad+ n(H(X)-\Delta_A+\delta(\epsilon)).
\end{align*}
Taking $n\rightarrow\infty$ completes the proof.
\end{proof}
\newpage

\bibliographystyle{IEEEtran.bst}
\bibliography{myBib}

\begin{thebibliography}{10}
\providecommand{\url}[1]{#1}
\csname url@samestyle\endcsname
\providecommand{\newblock}{\relax}
\providecommand{\bibinfo}[2]{#2}
\providecommand{\BIBentrySTDinterwordspacing}{\spaceskip=0pt\relax}
\providecommand{\BIBentryALTinterwordstretchfactor}{4}
\providecommand{\BIBentryALTinterwordspacing}{\spaceskip=\fontdimen2\font plus
\BIBentryALTinterwordstretchfactor\fontdimen3\font minus
  \fontdimen4\font\relax}
\providecommand{\BIBforeignlanguage}[2]{{%
\expandafter\ifx\csname l@#1\endcsname\relax
\typeout{** WARNING: IEEEtran.bst: No hyphenation pattern has been}%
\typeout{** loaded for the language `#1'. Using the pattern for}%
\typeout{** the default language instead.}%
\else
\language=\csname l@#1\endcsname
\fi
#2}}
\providecommand{\BIBdecl}{\relax}
\BIBdecl

\bibitem{bib:AhlswedeKorner1975}
R.~Ahlswede and J.~Korner, ``Source coding with side information and a converse
  for degraded broadcast channels,'' \emph{Information Theory, IEEE
  Transactions on}, vol.~21, no.~6, pp. 629 -- 637, Nov 1975.

\bibitem{bib:Wyner1975}
A.~Wyner, ``On source coding with side information at the decoder,'' \emph{Inf.
  Theory, IEEE Trans. on}, vol.~21, no.~3, pp. 294 -- 300, May 1975.

\bibitem{bib:SankarPrivacy2011}
L.~Sankar, S.~R. Rajagopalan, and H.~V. Poor, ``A theory of privacy and utility
  in databases,'' \emph{CoRR}, vol. abs/1102.3751v1, 2011.

\bibitem{bib:YHKimSutivongCover2008}
Y.-H. Kim, A.~Sutivong, and T.~Cover, ``State amplification,''
  \emph{Information Theory, IEEE Transactions on}, vol.~54, no.~5, pp. 1850
  --1859, may 2008.

\bibitem{bib:MerhavShamai2007}
N.~Merhav and S.~Shamai, ``Information rates subject to state masking,''
  \emph{Inf. Theory, IEEE Trans. on}, vol.~53, no.~6, pp. 2254 --2261, june
  2007.

\bibitem{bib:CsiszarKorner1981}
I.~Csiszar and J.~Korner, \emph{Information Theory: Coding Theorems for
  Discrete Memoryless Systems}.\hskip 1em plus 0.5em minus 0.4em\relax New
  York: Academic Press, 1981.

\bibitem{bib:KornerMarton1979}
J.~Korner and K.~Marton, ``How to encode the modulo-two sum of binary sources
  (corresp.),'' \emph{Information Theory, IEEE Transactions on}, vol.~25,
  no.~2, pp. 219 -- 221, Mar 1979.

\bibitem{bib:BergerLongo1977}
T.~Berger, \emph{Multiterminal Source Coding. In G. Longo (Ed.), The
  Information Theory Approach to Comms}.\hskip 1em plus 0.5em minus 0.4em\relax
  New York: Springer-Verlag, 1977.

\bibitem{bib:CourtadeWeissman2011}
T.~A. Courtade and T.~Weissman, ``Multiterminal source coding under logarithmic
  loss,'' \emph{CoRR}, vol. abs/1110.3069v2, 2011.

\bibitem{bib:ElGamalYHKim2012}
A.~El~Gamal and Y.-H. Kim, \emph{Network Information Theory}.\hskip 1em plus
  0.5em minus 0.4em\relax Cambridge University Press, 2012.

\end{thebibliography}

\end{document}